\documentclass[letterpaper, 10 pt, journal, twoside]{IEEEtran}

\IEEEoverridecommandlockouts                              

\usepackage{amsmath} 
\usepackage{amssymb}  
\usepackage{amsthm}

\title{
Exploiting symmetry for discrete-time reachability computations
}

\author{John Maidens and Murat Arcak
\thanks{This work was supported in part by the National Science Foundation under grant ECCS-1405413.}
\thanks{J. Maidens and M. Arcak are with the Department of Electrical Engineering \& Computer Sciences, University of California, Berkeley, 253 Cory Hall,
Berkeley, CA, 94720  USA,  e-mail: {\tt\small \{maidens, arcak\}@eecs.berkeley.edu}}%
}
\usepackage{comment}
\usepackage{amsfonts}
\usepackage{amsmath}
\usepackage{graphicx} 
\usepackage{algpseudocode}
\usepackage{algorithmicx}
\usepackage{algorithm}
\usepackage{url}       
\usepackage{caption}
\usepackage{subcaption}
\usepackage{fullpage}

                               \newcommand{\G}{\mathcal{G}}

                               \newcommand{\U}{\mathcal{U}}

                               \newcommand{\R}{\mathbb{R}}
                               
                               \newcommand{\W}{\mathcal{W}}
                               \newcommand{\X}{\mathcal{X}}

\newtheorem{prop}{Proposition}

\newtheorem{defn}{Definition}

\begin{document}

\maketitle

\begin{abstract}
We present a method of computing backward reachable sets for nonlinear discrete-time control systems possessing continuous symmetries. The starting point is a dynamic game formulation of reachability analysis where control inputs aim to maintain the state variables within a target tube despite disturbances. Our method exploits symmetry to compute the reachable sets in a lower-dimensional space, enabling a significant computational speedup. To achieve this, we present a general method for symmetry reduction based on the Cartan frame, which simplifies the dynamic programming iteration without algebraic manipulation of the state update equations. We illustrate the results by computing a backward reachable set for a six-dimensional reach-avoid game of two Dubins vehicles. 
\end{abstract}


\section{Introduction}

\IEEEPARstart{T}{he} computation of reachable sets has long played an important role in control theory \cite{Lee67, Leitmann82, Kurzhanskiy08}. 
Reachable sets appear in model-based safety verification of dynamic systems \cite{Lygeros99, Tomlin03} where reachability analysis proceeds by either demonstrating that any trajectory of a system model remains within a set of states labeled safe, or providing an example of a state trajectory that leaves the set of safe states. In particular, the computation of backward reachable sets is used to determine the set of states that can be restricted to the safe region via an appropriate control input \cite{Mitchell07}. The computation of reachable sets is also important for computing finite-state abstractions of continuous-state systems, which allow techniques from formal methods and model checking to be applied for automated verification and control synthesis \cite{Tabuada09}.  

One of the major challenges of reachability analysis is the computational cost of solving a dynamic programming recursion on a state space grid. Since the number of grid points increases exponentially with the state dimension, it becomes intractable to compute reachable sets for high-dimensional systems. A number of methods have been developed to address this challenge including projection-based methods \cite{Mitchell03}, 
convex relaxations based on occupation measures \cite{Lasserre08, Shia14}, methods exploiting monotone systems properties \cite{Coogan15}, simulation-based methods \cite{Donze07, Julius09, Huang12, Maidens15}  and methods based on support functions \cite{LeGuernic10, Frehse11}. 

The paper \cite{Chen16} is similar in spirit to the present paper in that it attempts to address the curse of dimensionality by describing reachable sets of high dimensional systems in terms of the reachable sets of lower dimensional systems. However it differs from the present work because it focuses on decoupled systems rather than symmetric systems and it is formulated for continuous-time systems which are not considered in this paper. 

Systems that possess symmetries are amenable to model order reduction techniques that simplify their analysis. Such techniques have been successfully applied in many aspects of control engineering including controllability for multi-agent systems \cite{Rahmani09}, stability analysis of networked systems \cite{Rufino17}, 
and control of mechanical systems \cite{Bloch96, Bullo99}. 

In this paper, we demonstrate that symmetries of a control system can be exploited to reduce the dimension of the state space for backward reachability computations. We build on results presented in \cite{Maidens17a, Maidens17b} where we have shown how optimal control policies can be efficiently computed via symmetry reduction. We will show that by exploiting symmetry to reduce the state space, we can speed up backward reachability computations by several orders of magnitude. The main results are the proof of two propositions that 1) establish that symmetries of the system dynamics and target tube imply symmetries of the corresponding effective target sets, and 2) provide a dynamic programming algorithm for computing the backward reachable sets over a reduced state space. 

Related results have appeared previously in other contexts including symmetry reduction for optimal control of nonlinear systems \cite{Grizzle84, Ohsawa13} and Markov decision processes \cite{Zinkevich01, Narayanamurthy07}. Dimensionality reduction using symmetry has also been applied in \cite{Mitchell05} to compute a lower-dimensional model for a two aircraft collision avoidance problem similar to the problem we present in Section \ref{sec:Dubins}. In contrast with previous work in reachability, here we present a general method for computing such symmetry reductions based on Cartan's method of moving frames, which allows us to compute a set of invariants of the dynamics using only information about the symmetries they possess. Thus, our method does not require the explicit algebraic computation of a lower-dimensional system model, which could facilitate its application in circumstances where this is difficult. 

Our method only requires the ability to evaluate the state update equations and to verify that they satisfy the symmetry properties given in Definition \ref{def:system_symmetries} but requires no explicit reduced model. This increases the ease with which the method can be applied to new systems and introduces the possibility of developing software packages generally applicable to computing reachable sets for systems with symmetries. 

We begin by presenting the main results on symmetry reduction for backward reachable set computation in Section \ref{sec:main_results}. We then apply these results to a dynamic pursuit-evasion game of two Dubins vehicles in Section \ref{sec:Dubins}. Conclusions and future directions for research are presented in Section \ref{sec:conclusion}. Software to reproduce the computational results presented in this paper is available at \url{https://github.com/maidens/2017-LCSS}.

\section{Symmetry reduction for discrete-time backward reachability}
\label{sec:main_results}

We begin by formulating a discrete-time backward reachability problem following the notation described in \cite{Bertsekas05}. We consider the system
\[
	x_{k+1} = f(x_k, u_k, w_k)
\]
where $x_k \in \X$ denotes the system's state at time step $k$, $u_k$ represents a control input to the system and $w_k$ represents a disturbance. We assume the control input $u_k$ is allowed to take values in a set $\U$ and the disturbance takes values in a set $\W$. The state transition map $f: \X \times \U \times \W \to \X$ defines the system's dynamics. 

We wish to choose a control policy $\pi = \{ \mu_0, \dots, \mu_{n-1} \}$ with $\mu_k: \X \to \U$ such that for each $k$ the state $x_k$ remains in a given set $X_k$, called the target set at time $k$, for any admissible sequence of disturbances $w_k$. Together the sets $\{X_k : k = 0, \dots, N\}$ form a target tube. The goal is to compute a sequence of effective target  sets $\bar X_k$ such that for any state $x_k \in \bar X_k$ there is a control policy $\pi$ such that $x_t \in X_t$ for all $t = k, \dots, N$. Note that if $X_k = \X$ for all $k \ne N$, this is equivalent to finding a backward reachable set \cite{Mitchell07} from the terminal set $X_N$. If $X_k = X_N$ for all $k$, this is equivalent to finding the viability kernel (in the case with no disturbance) \cite{SaintPierre94} or discriminating kernel \cite{Cardaliaguet94} of $X_N$. In the case where there is no control input, this is equivalent to finding the states for which there is no disturbance in $\W^N$ that can lead the state out of the target tube. 

By introducing stage costs 
\[
g_k(x_k) = \left\{ \begin{array}{ll} 0 & \text{ if } x_k \in X_k \\ 
1 & \text{ if } x_k \not\in X_k \end{array}\right. 
\]
we rewrite the problem of keeping the state within the target tube as a minimax optimal control problem 
\[
\begin{split}
 \operatorname*{\textbf{minimize }}_{ \pi }  & \ \ \ \max_{w\in \W^N}  J(\pi, w, x_0) \\
\textbf{subject to } & \ \ \ J(\pi, w, x_0) = \sum_{k = 0}^N g_k(x_k) \\
& \ \ \ x_{k+1} = f(x_k, \mu_k(x_k), w_k)
\end{split}
\]
where the minimization occurs over all admissible policies $\pi$. For a particular value of $x_0$, if the optimal value of the problem is 0 then $x_0 \in \bar X_0$ and if the optimal value is 1 then $x_0 \not \in \bar X_0$. 
A solution to this problem can then be found by computing a sequence of value functions $J_k : \X \to \R$ via the dynamic programming recursion 
\[
\begin{split}
J_N(x_N) &= g_N(x_N) \\ 
J_k(x_k) &= \min_{u_k \in \U} \max_{w_k \in \W} [g_k(x_k) + J_{k+1}(f(x_k, u_k, w_k))] 
\end{split}
\]
The effective target sets can then be described according to $\bar X_k = \{x_k : J_k(x_k) = 0 \}$.

\subsection{Reachable sets in symmetric systems} 

We now formalize what it means for a backward reachability problem to be symmetric via a transformation group acting on the system's states, inputs and disturbances. 

\begin{defn}
A transformation group on $\X \times \U \times \W$ is a set of tuples $h_\alpha = (\phi_\alpha, \chi_\alpha, \psi_\alpha)$ parametrized by elements of a group $\G$, such that $\phi_\alpha: \X \to \X$, $\chi_\alpha: \U \to \U$ and, $\psi_\alpha: \W \to \W$ are all bijections satisfying
\begin{itemize}
\item $\phi_e(x) = x$, $\chi_e(u) = u$, $\psi_e(w) = w$ when $e$ is the identity of the group $\mathcal{G}$ and
\item $\phi_{a * b}(x) = \phi_a \circ \phi_b (x)$, $\chi_{a * b}(u) = \chi_a \circ \chi_b (u)$, $\psi_{a * b}(x) = \psi_a \circ \psi_b (x)$ for all $a, b \in \mathcal{G}$ where $*$ denotes the group operation and $\circ$ denotes function composition. 
\end{itemize} 
\end{defn} 

\begin{defn}
We say that the backward reachability problem defined by $(f, X)$, where $X = (X_0, \dots, X_N)$, is invariant under the transformation group $h$ if for all $\alpha \in \G$ 
\begin{itemize}
\item $\phi_\alpha(x_k) \in X_k$ for all $x_k \in X_k$,  $k = 0, \dots, N$ and
\item $\phi^{-1}_\alpha \circ f_k( \phi_\alpha(x_k), \chi_\alpha(u_k), \psi_\alpha(w_k))$ \newline $= f_k(x_k, u_k, w_k), 
   \quad k = 0, \dots, N-1$. 
\end{itemize}
\label{def:system_symmetries}
\end{defn}

For reachability problems possessing symmetries, the symmetry extends to the effective target sets in a natural manner, as demonstrated by the following result. 

\begin{prop}
Let the reachability problem defined by $(f, X)$ be invariant under $h$. If $x_k \in \bar X_k$ then for all $\alpha \in \G$, we have $\phi_\alpha(x_k) \in \bar X_k$. That is, the reachable set is symmetric. 
\label{prop:symmetric_reachable_set}
\end{prop}
\begin{proof}
We prove this claim by induction on $k$. Note that the claim holds for $k=N$ by assumption since $\bar X_N = X_N$. Now, suppose that it holds for $k+1$. Thus for any $\alpha \in \G$ we have $J_{k+1} \circ \phi_\alpha = J_{k+1}$. Let $\alpha \in \G$ and $x_k \in \bar X_k$. We have 
\[
\begin{split}
	J_k(x_k) &= \min_{u_k \in \U} \max_{w_k \in \W} [g_k(x_k) + J_{k+1}(f(x_k, u_k, w_k))] \\
			 &= \min_{u_k \in \U} \max_{w_k \in \W} [g_k(\phi_\alpha(x_k))  \\ 
			 & \quad + J_{k+1}\circ \phi^{-1}_\alpha (f(\phi_\alpha(x_k), \chi_\alpha(u_k), \psi_\alpha(w_k)))] \\
			 &= \min_{\tilde u_k \in \U} \max_{\tilde w_k \in \W} [g_k(\phi_\alpha(x_k)) \\
			 & \quad + J_{k+1} \circ \phi_{\alpha^{-1}} (f(\phi_\alpha(x_k), \tilde u_k, \tilde w_k))] \\
			 &= \min_{\tilde u_k \in \U} \max_{\tilde w_k \in \W} [g_k(\phi_\alpha(x_k)) \\
			 & \quad + J_{k+1} (f(\phi_\alpha(x_k), \tilde u_k, \tilde w_k))] \\
			 &= J_k(\phi_\alpha(x_k)). 
\end{split}
\]
Thus $\phi_\alpha(x_k) \in \bar X_k$.
\end{proof}

\subsection{Cartan's moving frame method} 
\label{sec:Cartan}
The symmetry property of reachable sets derived in Proposition \ref{prop:symmetric_reachable_set} can be exploited to improve the efficiency of backward reachable set algorithms. To do so for continuous transformation groups (\textit{i.e.} when $\G$ is a Lie group), we rely on a formalism based on the moving frame method of Cartan \cite{Cartan37}, which we briefly introduce in this section following the notation of \cite{Bonnabel08}. 

We assume that $\G$ is an $r$-dimensional Lie group (with $r \le n$) acting on $\mathcal{X}$ via the diffeomorphisms $(\phi_\alpha)_{\alpha \in \mathcal{G}}$. We then split $\phi_\alpha$ as $(\phi_\alpha^a, \phi_\alpha^b)$ with $r$ and $n-r$ components respectively so that $\phi_\alpha^a$ is invertible.  For some $c$ in the range of $\phi^a$, we can then define a coordinate cross section to the orbits $\mathcal{C} = \{x: \phi_e^a(x) = c\}$ where $e$ is the group identity. This cross section is an $n-r$-dimensional submanifold of $\mathcal{X}$. Assume that for any point  $x \in \mathcal{X}$, there is a unique group element $\alpha \in \mathcal{G}$ such that $\phi_\alpha(x) \in \mathcal{C}$. If we denote $\alpha$ as $\gamma(x)$, then the map $\gamma: \mathcal{X} \to \mathcal{G}$ is called a moving frame for the symmetric system.

The moving frame can be found by solving the normalization equations:
\[
\phi_{\gamma(x)}^a(x) = c. 
\]
We then define a map $\rho: \mathcal{X} \to \mathbb{R}_{n-r}$ as 
\[
\rho(x) = \phi^b_{\gamma(x)}(x). 
\]
Note that, for all $\alpha \in \mathcal G$ we have $\rho(\phi_\alpha(x))=\rho(x)$ (see Section II.C.1 of \cite{Bonnabel08} for a proof). Thus, $\rho$ is  invariant to the action of $\G$ on the state space. Further, the restriction $\bar \rho$ of $\rho$ to $\mathcal{C}$ is injective, and thus has a well-defined inverse on its range. Thus we can use $\rho$ to define invariant coordinates on $\X$. 

In general, the theory of moving frames only guarantees that these invariants exist locally. However, for many problems of practical interest, including the example we will present in Section \ref{sec:moving_frame_example}, the local invariants can be extended globally. Thus we will present our results assuming a global set of invariants $\rho$ to simplify the notation.

\subsection{Illustrative example: Moving frame for a two-vehicle control problem}
\label{sec:moving_frame_example}

%
 Consider a six-dimensional state space describing a two vehicle system illustrated in Figure \ref{fig:vehicles_competitive}, with states modelled by the variables
\[
\begin{split}
x_k = \begin{bmatrix} z_k & y_k & \theta_k & \tilde z_k & \tilde y_k & \tilde \theta_k \end{bmatrix}^T. 
\end{split} 
\]
Define the rotation matrix
\[
R_\varphi = \begin{bmatrix} 
\cos \varphi & -\sin \varphi & 0 & 0 & 0 & 0 \\
\sin \varphi & \cos \varphi & 0 & 0 & 0 & 0 \\
0 & 0 & 1 & 0 & 0 & 0 \\
0 & 0 & 0 & \cos \varphi & -\sin \varphi & 0 \\
0 & 0 & 0 & \sin \varphi & \cos \varphi & 0 \\
0 & 0 & 0 & 0 & 0 & 1 
\end{bmatrix}. 
\]
If we define coordinates on the symmetry group $\mathcal{G} = SE(2)$ in terms of a rotation angle $\theta'$ and a translation by $(z', y')$ then this system's dynamics are invariant under the transformation group $h_\alpha = (\phi_\alpha, \chi_\alpha, \psi_\alpha)$ where 
\[
\begin{split}
\phi_{(z', y', \theta')}(x) &= R_{\theta'} x + \begin{bmatrix} z', & y', & \theta', &  z', &  y', & \theta' \end{bmatrix}^T \\
\chi_{(z', y', \theta')}(u) &= u \\
\psi_{(z', y', \theta')}(w) &= w. 
\end{split} 
\]

For this system, the moving frame $\gamma$ can be computed by solving the normalization equations
\[
0 = \phi^a_{\gamma} (x) = \begin{bmatrix} \cos \gamma_3 & -\sin \gamma_3 & 0 \\ \sin \gamma_3 & \cos \gamma_3 & 0 \\ 0 & 0 & 1 \end{bmatrix} 
\begin{bmatrix} x_1 \\ x_2 \\ x_3 \end{bmatrix} + \begin{bmatrix} \gamma_1 \\ \gamma_2 \\ \gamma_3 \end{bmatrix} 
\]
which give
\[
\begin{split}
\gamma(x) 
&= - \begin{bmatrix} \cos(-x_3) & -\sin(-x_3) & 0 \\ \sin(-x_3) & \cos(-x_3) & 0 \\ 0 & 0 & 1 \end{bmatrix}
 \begin{bmatrix} x_1 \\ x_2 \\ x_3 \end{bmatrix} \\
&= \begin{bmatrix} -x_1 \cos x_3 - x_2 \sin x_3 \\
                    x_1 \sin x_3 - x_2 \cos x_3 \\
                    -x_3 \end{bmatrix}.  
\end{split} 
\]
Three invariants can then be computed as
\[
\resizebox{\columnwidth}{!}{$
\begin{aligned}
\rho(x) &= \phi^b_{\gamma(x)} (x) \\
&= \begin{bmatrix} \cos(\gamma_3(x)) & -\sin(\gamma_3(x)) & 0 \\ \sin(\gamma_3(x)) & \cos(\gamma_3(x)) & 0 \\ 0 & 0 & 1 \end{bmatrix}
 \begin{bmatrix} x_4 \\ x_5 \\ x_6 \end{bmatrix} + \begin{bmatrix} \gamma_1(x) \\ \gamma_2(x) \\ \gamma_3(x) \end{bmatrix} \\
&= \begin{bmatrix} 
(x_4 - x_1) \cos x_3 + (x_5 - x_2) \sin x_3 \\
-(x_4 - x_1) \sin x_3 + (x_5 - x_2) \cos x_3 \\
x_6 - x_3 
\end{bmatrix}.  
\end{aligned}
$}
\]
Restricted to the cross-section \newline $\{x : 0 = \phi_e^a(x) = \begin{bmatrix}x_1 & x_2 & x_3 \end{bmatrix}^T \}$, $\rho$ is injective, with inverse given by 
\[
\bar \rho^{-1} (\bar x) = \begin{bmatrix} 0, & 0, & 0, & \bar x_1, & \bar x_2, & \bar x_3 \end{bmatrix}^T.  
\]

The reduced state space for this two vehicle system can be interpreted as describing the relative orientation of the two vehicles, by introducing a moving coordinate system that fixes one vehicle at the origin. Note that the reduction functions $\rho$ and $\bar \rho^{-1}$ are computed based only on the symmetries of the system, and do not depend on the system dynamics. This is in contrast with more \textit{ad hoc} approaches where the dynamics on a lower dimensional space are computed algebraically. 

\subsection{An efficient algorithm for backward reachable set computation in symmetric systems} 

The invariance properties of reachable sets established in Proposition  \ref{prop:symmetric_reachable_set} suggest that this property might be exploited to reduce the dimension of the space in which dynamic programming is performed to compute backward reachable sets.  This can be done by defining a reduced value function $\bar J_k(\bar x_k) = J_k( \bar \rho^{-1}(\bar x_k))$ taking values on a space of dimension $n-r$. The following result establishes that $\bar J_k$ can be computed via a dynamic programming iteration based on the invariants $\rho$ of the Cartan frame introduced in the previous section. 

\begin{prop}
\label{prop:reduced_dynamic_programming}
The reachable set can be computed via a dynamic programming iteration in reduced coordinates:
\[
\begin{split}
\bar J_N(\bar x_N) &= g_N(\bar \rho^{-1}(\bar x_N)) \\ 
\bar J_k(\bar x_k) &= \min_{u_k \in \U} \max_{w_k \in \W} g_k(\bar \rho^{-1}( \bar x_k)) \\
 & \quad + \bar J_{k+1}(\rho(f( \bar \rho^{-1}(\bar x_k), u_k, w_k))). 
\end{split}
\]
\end{prop}
\noindent\textit{Proof.}
Note that from the invariance property of $J_{k+1}$ we have that for any $k$ and any $x \in \X$, 
\[
\begin{split}
J_{k+1}(x) &= J_{k+1}(\phi_{\gamma(x)}(x)) = J_{k+1}(\bar \rho^{-1} \circ \bar \rho \circ \phi_{\gamma(x)}(x)) \\
		   &= J_{k+1}(\bar \rho^{-1} \circ \rho(x)) = \bar J_{k+1} (\rho(x)). 
\end{split}
\]
Thus we have 
\[
\resizebox{\columnwidth}{!}{$
\begin{aligned}
\bar J_k(\bar x_k) &= J_k( \bar \rho^{-1}(\bar x_k)) \\
                   &= \min_{u_k \in \U} \max_{w_k \in \W} [g_k(\bar \rho^{-1}(\bar x_k)) 
                   + J_{k+1}(f(\bar \rho^{-1}(\bar x_k), u_k, w_k))] \\
                   &= \min_{u_k \in \U} \max_{w_k \in \W} [g_k(\bar \rho^{-1}(\bar x_k)) 
                   + \bar J_{k+1}(\rho(f(\bar \rho^{-1}(\bar x_k), u_k, w_k)))]. \qed
\end{aligned}
$}
\]

Note that in this recursion $\bar x_k$ parametrizes a space of lower dimension than $\mathcal{X}$. Thus the dynamic programming iteration can be performed much more efficiently. Further, the dynamic programming iteration can be performed for general state update maps $f$ without any knowledge of $f$ beyond being able to evaluate it and verify its symmetries. 

Once the reduced costs $\bar J_k$ are computed, effective target sets for the original system are defined implicitly via
\[
\bar X_k = \{x \in \mathcal{X} : \bar J_k(\rho(x)) = J_k(x) = 0 \}
\]
and a safely-preserving control policy on the original state space is defined via
\[
\mu_k = \bar \mu_k \circ \rho
\]
where $\bar \mu_k$ is the policy computed via the reduced dynamic programming iteration.

\section{Application: Reachability problem with two Dubins vehicles} 
\label{sec:Dubins}

In this section, we demonstrate how these results can be applied to a reach-avoid game of two identical Dubins vehicles, as depicted in Figure \ref{fig:vehicles_competitive}. Vehicle 1 wishes to reach a configuration in which it can view vehicle 2 using a forward-facing camera mounted on its hood, whereas vehicle 2 wishes to avoid reaching such a configuration where it has been detected. Thus, vehicle 2 wishes to remain outside the shaded region in Figure \ref{fig:vehicles_competitive} at all times. 

The state of vehicle 1 is modelled by variables $(z, y, \theta)$ representing a two-dimensional position $(z, y) \in \R^2$ along with a heading $\theta \in [0, 2\pi)$ while the state of vehicle 2 is modelled by variables with the same interpretation denoted $(\tilde z, \tilde y, \tilde \theta)$.  
The dynamics of each vehicle are governed by the equations
\[
\begin{split}
z_{k+1} &= z_k + v_k \cos(\theta_k) \\
y_{k+1} &= y_k + v_k \sin(\theta_k) \\
\theta_{k+1} &= \theta_k + \frac{1}{L} v_k \sin s_k
\end{split}
\]
where $v_k$ describes a velocity input, $s_k$ describes a steering angle input and $L$ is a parameter that determines the vehicle's turning radius. 

This problem exhibits symmetries corresponding to the rigid motions in two-dimensional Euclidean space, that is, the symmetry group is the three-dimensional group $\mathcal{G} = SE(2)$.
 Let us denote the full state, input, and disturbance vectors of the system as 
\[
\resizebox{\columnwidth}{!}{$
x_k = \begin{bmatrix} z_k & y_k & \theta_k & \tilde z_k & \tilde y_k & \tilde \theta_k \end{bmatrix}^T,   \quad 
u_k = \begin{bmatrix}  \tilde v_k \\ \tilde s_k \end{bmatrix} 
w_k = \begin{bmatrix} v_k \\ s_k \end{bmatrix}.
$}
\]
Using these coordinates, the objectives of the two vehicles can be formulated via the cost function 
\[
g_k(x, y, \theta, \tilde x, \tilde y, \tilde \theta) 
= \begin{cases}
1 & \text{ if } 
\resizebox{.4\columnwidth}{!}{$
\begin{array}{c}
(\tilde z - z)^2 + (\tilde y - y)^2 \le r^2 \\
\text{ and } \\
\frac{(\tilde z - z)\cos\theta + (\tilde y - y)\sin\theta}{\sqrt{(\tilde z - z)^2 + (\tilde y - y)^2}} \ge \cos\epsilon \end{array} 
$}
\\
0 & \text{ otherwise.}
\end{cases}
\] 
\begin{figure}[ht!]
\centering
\includegraphics[width=0.75\columnwidth]{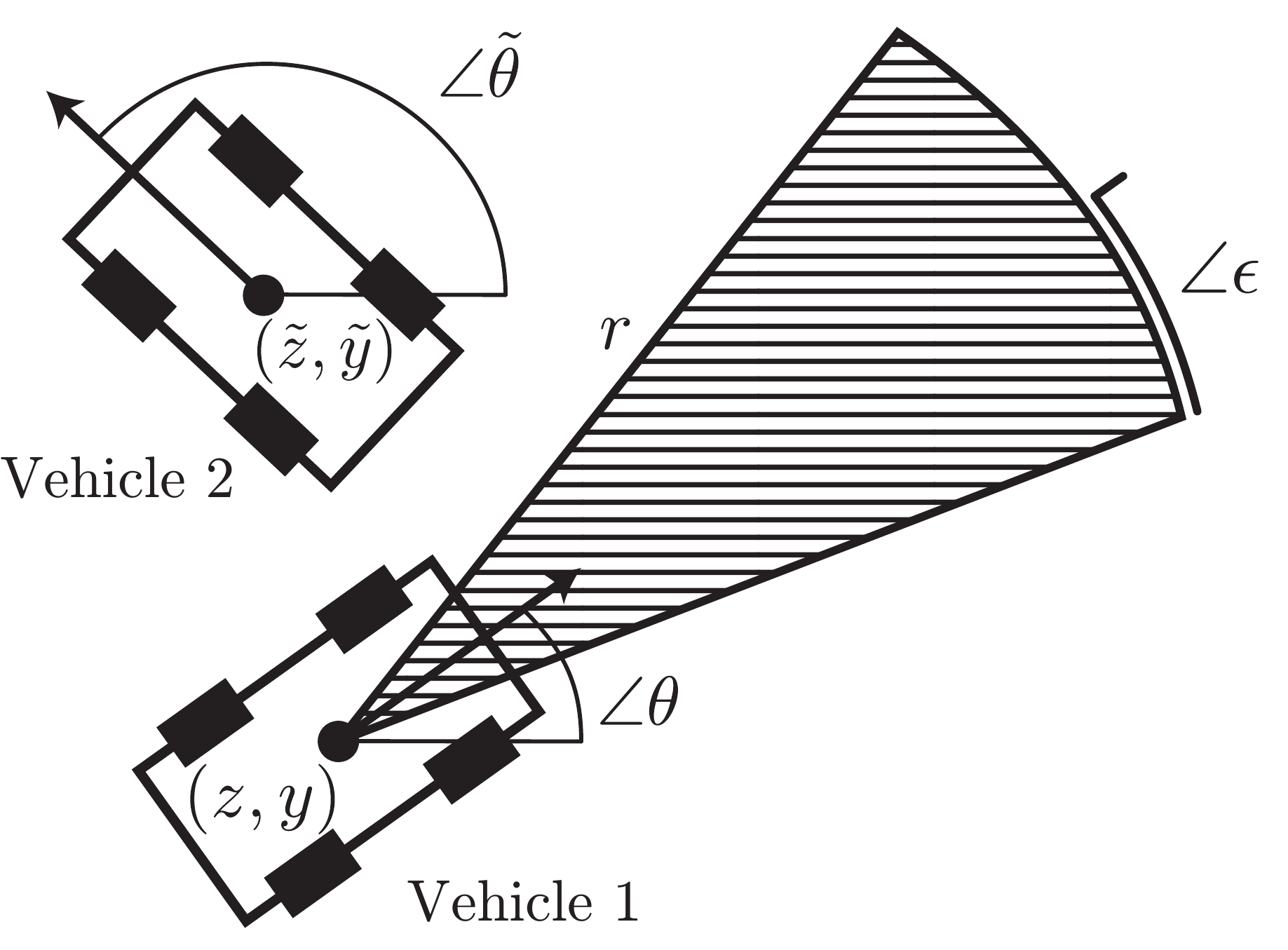}
\caption{Illustration of the two vehicle reachability problem}
\label{fig:vehicles_competitive}
\end{figure}

Using the theory developed in this paper, we demonstrate how the reachable set for a six-dimensional model of this system can be computed by gridding a reduced state space of dimension three.

\subsection{Simulations} 
Using $\rho$ and $\bar \rho^{-1}$ computed in Section \ref{sec:moving_frame_example}, we are now able to compute the backward reachable sets via the recursion given in Proposition \ref{prop:reduced_dynamic_programming}. Inputs are assumed to be constrained to the sets $v_k \in \{0, V_{max}\}$ and $s_k \in \{-S_{max}, 0, S_{max} \}$. We compute the backward reachable set for this system with parameter values  $L = 1$, $V_{max} = 0.05$, $S_{max} = 1$, $\epsilon = 15^\circ$ and $r = 0.5$ over a horizon $N = 10$. The resulting reachable set, computed on a $51 \times 51 \times 51$ grid is shown in Figure \ref{fig:dubins_reachable_set}. Python software to reproduce these plots is available online at \newline \url{https://github.com/maidens/2017-LCSS}. 

\begin{figure}[ht!]
    \centering
    \begin{subfigure}[b]{0.45\columnwidth}
    	\centering
        \includegraphics[width=0.6\textwidth]{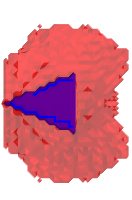}
        \caption{View from above in $(\bar x_1, \bar x_2)$ plane.}
        \label{fig:dubins_reachable_set_a}
    \end{subfigure}
    ~
    \begin{subfigure}[b]{0.45\columnwidth}
        \centering
        \includegraphics[width=0.6\textwidth]{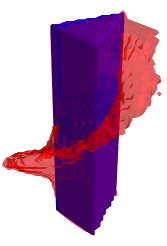}
        \caption{Three-dimensional view of $(\bar x_1, \bar x_2, \bar x_3)$.}
        \label{fig:dubins_reachable_set_b}
    \end{subfigure}
    \caption{Two views of the reachable set computed for the Dubins vehicle problem. The complement of the target set $X_k$ is shown in solid blue. The complement of the effective target set $\bar X_0$ is shown in transparent red.}
     \label{fig:dubins_reachable_set}
\end{figure}

The effective target set in the reduced space defined by $\rho$ can be interpreted in terms of a moving coordinate system where vehicle 1 is frozen at the origin. Thus this general method is able to automatically reproduce intuitive results similar to those resulting from a model derived by hand in \cite{Mitchell05}. Vehicle 2 has a winning strategy for avoiding the blue detection region whenever it begins outside the transparent region plotted in red. But if vehicle 2 begins inside the red region, vehicle 1 has a strategy that can force vehicle 2 into the blue region.

\subsection{Timing experiments}

To illustrate the computational savings that this technique provides, we compare the symmetry reduction approach with a baseline approach that does not exploit symmetry for computing the reachable sets using a varying number of grid points in each state dimension. The dynamic programming recursion is implemented in non-optimized python code run via the standard CPython interpreter on a laptop computer with 2.3 GHz Intel Core i7 processor and 8 GB memory. Wall time to compute the reachable set over a horizon of $N = 1$ is shown in Table \ref{tab:timing}. The results demonstrate that symmetry reduction can accelerate reachability computations by several orders of magnitude, as the exact reachability method used here scales exponentially with the state dimension. In future work we will investigate exploiting symmetry for approximate reachability methods that scale polynomially. 
\begin{table}[ht!]
\begin{center}
  \begin{tabular}{| l | c | c | c |}
    \hline
    Number of grid points in each dimension &  5 & 11 & 51 \\ \hline
    Wall time for reduced model (seconds) & 0.866 & 10.2 & 953 \\ \hline
    Wall time for baseline model (seconds) & 176 & * & * \\
    \hline
  \end{tabular}
\end{center}
\caption{Comparison of running time for symmetry-reduced reachable set computation compared against a non-reduced baseline across for various grid densities. The character $*$ denotes that the computation timed out after 7200 seconds.} 
\label{tab:timing}
\end{table}
\section{Conclusions}
\label{sec:conclusion}
We have presented a general method for reducing the complexity of backward reachable set computations for discrete-time systems with symmetries. This method can be used to substantially accelerate backward reachable set computations, and can be performed without any knowledge of the state update map beyond being able to evaluate it and verify its symmetries.

Interesting future research directions include extending these results to the continuous-time setting through the study of viscosity solutions of Hamilton-Jacobi-Isaacs partial differential equations, combining this technique with approximate reachability approaches to increase its scalability to high-dimensional systems, and developing numerical methods for automatically computing solutions to the normalization equation which would enable $\rho$ and $\bar \rho^{-1}$ to be computed automatically. 

\bibliographystyle{IEEEtran}
\bibliography{references}

\end{document}